\documentclass[pra,aps,twocolumn,10pt,showpacs,nofootinbib]{revtex4-2}
\usepackage[utf8]{inputenc}
\usepackage[T1]{fontenc}
\usepackage{amsmath,amsthm,charter,amssymb,amsfonts}
\usepackage{mathcmd}
\usepackage{graphicx} 

\usepackage{braket}
\usepackage{bm,bbold}
\usepackage{mathrsfs}
\usepackage{stmaryrd}
\usepackage{esint}
\usepackage{xcolor}
\usepackage[table]{xcolor}

\usepackage{tikz,rotating,float,qcircuit}
\usepackage{cancel}

\usepackage{ulem}
\usepackage[colorlinks=true, linkcolor=purple, citecolor=blue]{hyperref}

\newtheorem{definition}{Definition}
\newtheorem{proposition}{Proposition}
\newtheorem{theorem}{Theorem}

\newcommand{\id}{\mathrm{id}}
\newcommand{\Tr}{\mathrm{Tr}}
\newcommand{\cra}{\rightarrow_{\mathrm{c}}}

\begin{document}

    \title{Censorship of quantum resources against quantum catalysis}

    \author{Julien Pinske}
    \email{julien.pinske@nbi.ku.dk}
    \affiliation{Niels Bohr Institute, University of Copenhagen, Jagtvej 155 A, DK-2200 Copenhagen, Denmark}
    
    \author{Klaus M\o lmer}
    \affiliation{Niels Bohr Institute, University of Copenhagen, Jagtvej 155 A, DK-2200 Copenhagen, Denmark}

    \date{\today}

    \begin{abstract}
        In quantum censorship, an agency oversees quantum communication in a public-domain network. 
        The agency restricts the users communication to the free states of a quantum resource theory (QRT). 
        Despite quantum correlations being fragile, any realistic censorship leaves behind some quantumness, raising concerns that censorship may be overcome through revival or distillation of quantum resources.
        Here, we introduce censorship protocols that do not require a perfect erasure of a quantum resource, but rather deem censorship successful if users are unable to restore the original quantum state using free operations. 
        We investigate under which conditions censorship is secure, and when it might fail.
        Moreover, we address the issue of account sharing in quantum networks, wherein independent parties assist in transmitting quantum resources to censored users. 
        This connects resource censorship to timely topics such as quantum catalysis and resource-assisted communication.
        Censorship protocols offer a novel perspective on quantum network security, that differs fundamentally from existing approaches such as quantum and post-quantum cryptography.
    \end{abstract}
    
    \maketitle

    \section{Introduction}
    \label{sec:intro}

    The contemporary pursuit of storing, transmitting, and processing quantum information fuels the prospect of a widely accessible quantum internet \cite{K08,RH25}. 
    There, we may see commercial agencies offering access to quantum communication. 
    In such public-domain networks, access to certain quantum resources might be restricted, either for security, policy, or pricing reasons. 
    For example, priority members may be granted access to quantum resources such as entanglement or coherence, while regular users are limited to classical communication.

    In quantum censorship \cite{PS24,PM24}, a dominant agency restricts the users' communication to the free states of a quantum resource theory (QRT) \cite{CG19}. 
    Censorship is applied to each sender–receiver channel, ideally ensuring that only free states can be sent. 
    The key challenge is to suppress the transmission of resource states without disturbing communication with free states. 
    Various quantum resources can be censored \cite{PS24}, including coherence \cite{A06,BC14,LM14,SW18} and asymmetry \cite{BR07,GS08,CA25}. 
    Censorship protocols may also condition their action on classical information about a quantum message \cite{PM24}, enabling secure surveillance of imaginarity \cite{WK21,XG21} and entanglement \cite{HH09,VS14,PAH24}. 
    In contrast, quantum discord \cite{OZ01,MB12,WPM09} and non-Gaussianity \cite{ES02,WP12} are nonconvex resources, that is, they can be created from classical mixing across multiple transmissions, and are therefore not susceptible to censorship.
    Furthermore, there are resources that can be activated, that is, having many copies of free states constitutes a resource.
    This includes cases in which a desired quantum correlation can be concentrated or distilled \cite{R99}. 
    Examples of distillable resources are Bell nonlocality \cite{P12}, bi-separable but not fully-separable states \cite{PV22}, as well as magic states \cite{BK05,AC12}.
    These resources pose a challenge to a decentralized censorship, because it cannot be verified locally, whether a resource is transmitted.

    Another limitation of censorship protocols is their reliance on complete resource erasure, assuming that quantum states are fully converted into classical ones. 
    However, residual quantumness is often inevitable due to physical limitations such as a non-Markovian revival of coherences \cite{LP10,YZ20}, or the thermodynamic cost of resetting a quantum system to a classical register \cite{L61,GP05}.
    Moreover, requiring the complete removal of a quantum resource is an unreasonably stringent form of censorship.
    A more realistic approach would deem censorship secure if the users cannot make practical use of the residual quantum resources they are left with.

    In this work, we propose an operational approach to resource censorship. 
    Rather than requiring perfect filtering of a quantum resource, censorship is deemed successful if the users cannot restore the original resource state using free operations. 
    This weaker but operationally meaningful approach acknowledges that some quantum correlations may survive, yet remain inaccessible under the allowed operational constraints.
    We formalize this intuition by introducing resource-reducing channels: operations that leave free states unchanged while degrading the resource content of all other states. 
    Censorship protocols are built by applying such operations locally to each sender–receiver connection,
    ensuring that quantum messages are reduced to, or approach, classical ones.

    We analyze how users might attempt to circumvent censorship through collaboration.
    This includes multi-party attacks in which the censored users coordinate via joint operations to restore the pre-censorship state (Fig. \ref{fig:censor}).
    Additionally, we study the case of account sharing, where independent parties enable the censored users to perform quantum catalysis \cite{JP99,CC10}. 
    This connects resource censorship to timely topics in quantum information theory including, but not limited to, resource embezzlement \cite{DH03,TS22}, the simulation of quantum operations \cite{KE06,VC02}, and entanglement-assisted quantum communication \cite{BS99,B02,HD08}.
    We illustrate our framework for quantum coherence, where censorship limits communication to incoherent states. 
    We analyze several settings, demonstrating both secure and breakable censorship regimes.

    The article is structured as follows. 
    Section~\ref{sec:QRT} provides a review of QRT.
    There, we introduce resource-reducing channels, the key notion in resource censorship.
    In Sec.~\ref{sec:censor} protocols for censorship are devised.
    We study under which conditions the censorship is secure against attacks by collaborating users.
    Section~\ref{sec:account} considers scenarios that include resource sharing by independent parties.
    In Sec.~\ref{sec:coherence} and Sec.~\ref{sec:asymmetry}, we illustrate the theory for censorship of quantum coherence and quantum asymmetry, respectively. 
    Finally, Sec.~\ref{sec:Fin} is reserved for a summary of the article and concluding remarks.

    \begin{figure}[t]
        \centering
        \begin{tikzpicture}
        \node at (0,0) {\includegraphics[width=0.5\textwidth]{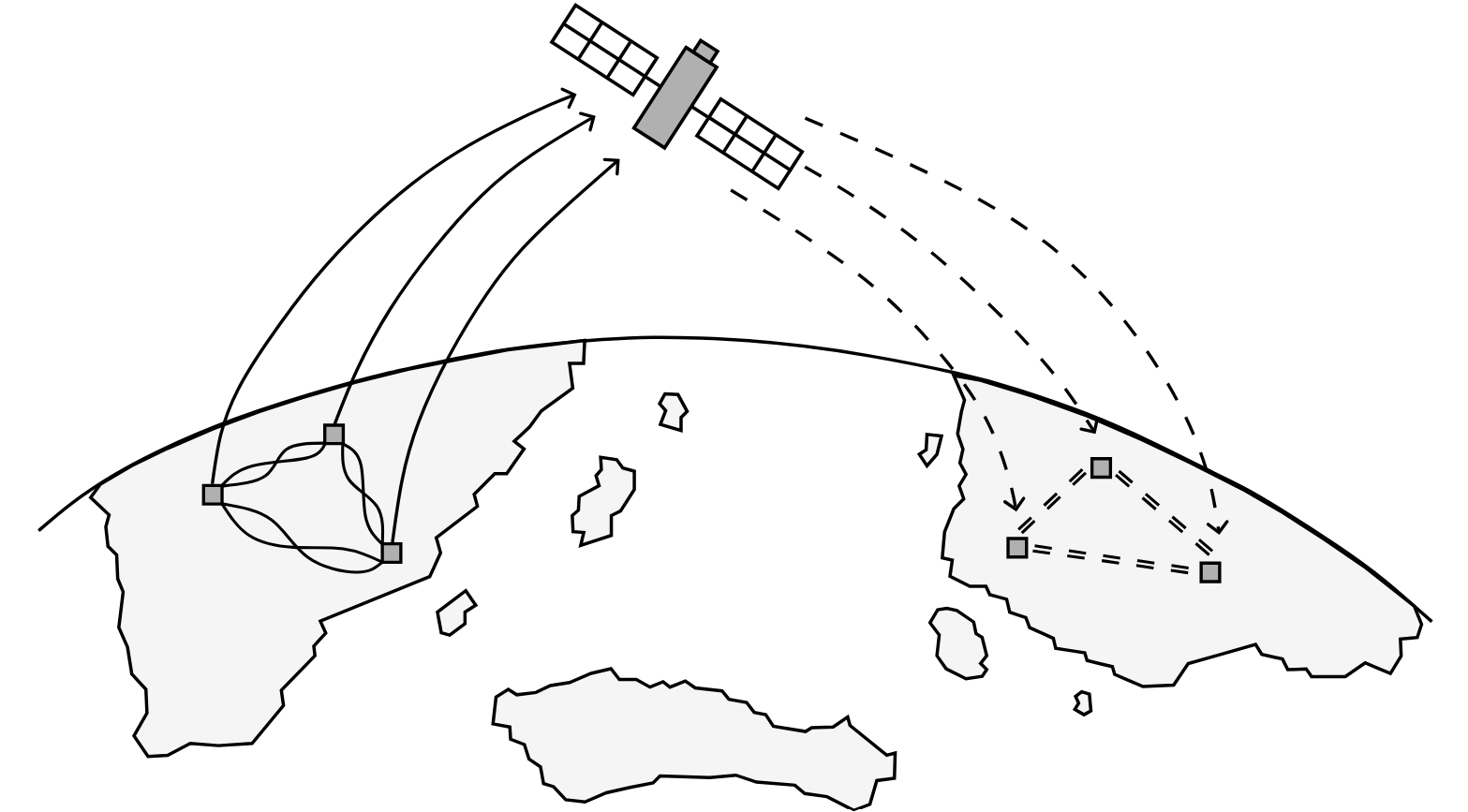}};
        \node[rotate=-28] at (-0.35,1.9) {\scriptsize{$\Omega$}};
        \node at (-3.4,-0.6) {\tiny{$A_1$}};
        \node at (-2.2,-0.1) {\tiny{$A_2$}};
        \node at (-1.9,-0.7) {\tiny{$A_3$}};
        \node at (1.51,-0.9) {\tiny{$B_1$}};
        \node at (2,-0.25) {\tiny{$B_2$}};
        \node at (2.97,-1.22) {\tiny{$B_3$}};
        \end{tikzpicture}
        \caption{\label{fig:censor}
        In a resource censorship, an agency oversees quantum communication in a network by applying a resource-reducing channel $\Omega$ before sending a quantum message. 
        To break the censorship, the senders combine their quantum resources (tangled lines) for improved transmission and the receivers coordinate their classical resources (dashed lines) for state recovery.}
    \end{figure}
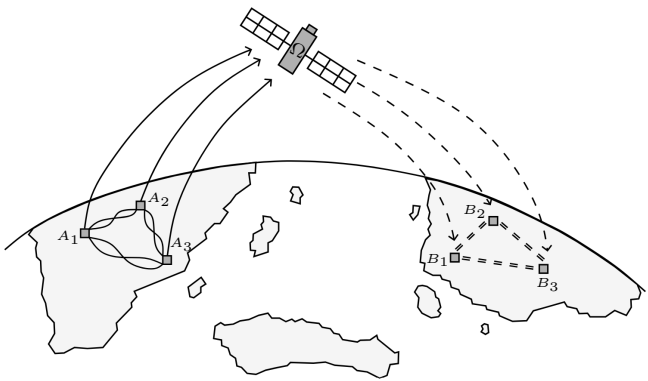

    \section{Quantum resource theories}
    \label{sec:QRT}

    When establishing a censorship on quantum information, we first divide resource states, whose distribution is to be prevented, from free states, which propagate through a network unchanged.
    This distinction is made within a quantum resource theory (QRT) \cite{CG19}.
    In a QRT one defines a set of free states $\mathcal{F}(A)$, being a subset of the set of quantum states $\mathcal{D}(A)$.
    The set $\mathcal{D}(A)$ contains positive(-semidefinite), unit-trace operators $\rho$ acting on the (here, finite-dimensional) Hilbert space $\mathcal{H}_A$ of a system $A$.
    If $\rho\notin\mathcal{F}(A)$ is not free, it is said to be a resource.

    Physical operations are mathematically expressed as quantum channels \cite{W18}, i.e., linear maps $\Lambda:\mathcal{D}(A)\to\mathcal{D}(B)$ that take quantum states to quantum states.
    Any QRT comes with a set of free operations $\mathcal{O}(A\to B)$, which do not create a resource:
    \begin{equation}
        \label{eq:res-non-gen}
        \forall \sigma\in\mathcal{F}(A):\quad \Lambda(\sigma)\in\mathcal{F}(B),
    \end{equation}
    for all $\Lambda\in\mathcal{O}(A\to B)$.
    The free operations contain the identity, $\id(\rho)=\rho$, and are closed under composition, i.e., $\Phi\circ \Lambda\in\mathcal{O}(A\to C)$, for all $\Lambda\in\mathcal{O}(A\to B)$ and $\Phi\in\mathcal{O}(B\to C)$ and any choice of systems $A$, $B$, and $C$.
    We write $\mathcal{O}(B)=\mathcal{O}(B\to B)$ for the free operations on $B$, and assume that preparing a free state, 
    \begin{equation}
        \label{eq:prep-free-state}
        \Lambda(\rho)=\Tr(\rho)\sigma~~\text{with}~~\sigma\in\mathcal{F}(B),
    \end{equation}
    is a free operation, i.e., $\Lambda\in\mathcal{O}(B)$.
    
    \subsection{Resource-reducing channels}
    \label{ssec:resource-reducing}

    We can compare two states based on the possibility of converting one into the other using free operations.
    \begin{definition}
        \normalfont
        \label{def:order}
        Let $\rho$ and $\omega$ be quantum states. 
        We write
        \begin{equation}
            \omega\rightarrow \rho,
        \end{equation}
        if there is $\Lambda\in\mathcal{O}$ such that $\rho=\Lambda(\omega)$.
        Otherwise, $\omega\nrightarrow\rho$.
    \end{definition}
    If $\omega$ can be transformed freely into $\rho$ but not vice versa, then $\rho$ is less resourceful than $\omega$.
    For example, any $\sigma\in\mathcal{F}$ is less resourceful than any $\rho\notin\mathcal{F}$, where
    $\sigma\nrightarrow \rho$ by Eq.~\eqref{eq:res-non-gen} and $\rho\rightarrow \sigma$ follows from Eq.~\eqref{eq:prep-free-state}.

    Next, we formally define the notion of resource-reducing channels, a key notion in resource censorship.
    \begin{definition}
        \label{def:RR}
        \normalfont
        A free operation $\Omega\in\mathcal{O}(A\to B)$ obeying
        \begin{align}
            \text{(i)\quad}
            \forall\rho\notin\mathcal F(A):& ~~\Omega(\rho)\nrightarrow\rho,
            \tag{resource-reducing}\\
            \text{(ii)\quad} \forall\sigma\in\mathcal F(A):& ~~\Omega(\sigma)=\sigma,\tag{freeness-preserving}
        \end{align}
        is called resource reducing.
    \end{definition}
    Simply speaking, $\Omega$ reduces quantum resources in a state, and leaves the input unchanged if it was already a free state. 
    For condition (ii) to be well defined, the input and output Hilbert space $\mathcal{H}_A$ and $\mathcal{H}_B$ need to have the same dimension. 
    Note that $\Omega$ is in general not unique; for any free unitary operation $\Phi\in\mathcal{O}(A)$, with $\Phi(\rho)=U\rho U^\dag$ and $U^\dag U = \mathbb{1}_A$, the channel
    \begin{equation}
        \Omega_{\Phi} = \Phi^{-1}\circ \Delta \circ \Phi
    \end{equation}
    is resource reducing as well. See Appendix~\ref{app:resource-reducing} for a proof.

    \subsection{Tensor-product structures}
    \label{ssec:TPS}
     
    Since censorship will be applied to quantum networks, one has to clarify what constitutes a resource on a joint system.
    Throughout, we consider QRTs that admit a tensor-product structure \cite{CG19}.
    Firstly, this means that the independent preparation of free states $\sigma_{A_1},\dots, \sigma_{A_N}$ yields a free state on the joint system, i.e.,
    \begin{equation}
        \label{eq:non-active}
        \sigma_{A_1}\otimes\dots \otimes \sigma_{A_N}\in\mathcal{F}(A_1\dots A_N).
    \end{equation}
    Secondly, discarding subsystems does not create a resource;
    i.e., for $\sigma\in\mathcal{F}(A_1\dots A_N)$, its marginals 
    \begin{equation}
        \mathrm{Tr}_{a}(\sigma)\in \mathcal{F}(A_1\dots A_{a-1}A_{a+1}\dots A_N)
    \end{equation}
    are free, too.
    Thirdly, parallel application of free operations $\Lambda_{A_1},\dots, \Lambda_{A_N}$ constitutes a free operation, i.e., 
    \begin{equation}
        \label{eq:para-free-ops}
        \Lambda_{A_1}\otimes\dots \otimes \Lambda_{A_N}\in\mathcal{O}(A_1\dots A_N\to B_1\dots B_N).
    \end{equation}

    \section{Quantum censorship}
    \label{sec:censor}

    Consider $N$ senders $A_1,\dots, A_N$ who have access to quantum resources.
    In an unregulated network, each sender is connected to one of the receivers $B_1,\dots, B_N$ via a (noiseless) communication channel. 
    However, to prevent the distribution of quantum resources, an agent sits in between each sender-receiver pair. 
    The agent’s goal is to limit the type of states that can be transmitted to the free states of a QRT.
    The agent informs the users that only communication with free states is authorized, constituting a user agreement. 
    To enforce that policy, the agent applies a resource-reducing channel $\Omega$. 
    The information-processing protocol for the censorship is
    \begin{equation*}
    	\Qcircuit @C=1em @R=.7em {
    		\lstick{A_1} & \qw &\gate{\Omega} & \qw & \qw & B_1\\
    		\lstick{\vdots} & &\vdots & & &\vdots\\
    		\lstick{} & & & & &\\
    		\lstick{A_N} & \qw & \gate{\Omega} & \qw & \qw & B_N.\\
    	}
    \end{equation*}

    If the input state $\sigma_{A_1}\otimes\dots\otimes \sigma_{A_N}$ is free, it remains unchanged by $\Omega^{\otimes N}$, because of (ii).
    This allows users to carry out (undisturbed) communication with free states.
    On the other hand, the users might try to collaborate to transmit quantum resources through the network.
    As a strategy to bypass the censorship, the senders prepare a joint resource state $\rho\notin \mathcal{F}(A_1\dots A_N)$.
    After the censorship leaves the receivers with the state $\Omega^{\otimes N}(\rho)$, their task is to restore $\rho$.
    Since the receivers do not have access to quantum resources, they can only apply a joint free operation $\Lambda\in\mathcal{O}(B_1\dots B_N)$. 
    The modified information-processing protocol is
    \begin{equation*}
    	\Qcircuit @C=1em @R=.7em {
    		\lstick{A_1} & \qw &\gate{\Omega} & \qw & \multigate{3}{\Lambda} & \qw & B_1\\
    		\lstick{\vdots} & &\vdots & & & &\vdots\\
    		\lstick{} & & & & & &\\
    		\lstick{A_N} & \qw &\gate{\Omega} & \qw & \ghost{\Lambda} & \qw & B_N.\\
    	}
    \end{equation*}
    
    If it is not possible to restore the original state $\rho$ in this way, then the censorship is termed secure.
    \begin{definition}
        \normalfont
        \label{def:secure}
        Censorship is secure if
        \begin{equation}
            \forall \rho\notin \mathcal{F}(A_1\dots A_N):\quad  \Omega^{\otimes N}(\rho)\nrightarrow\rho.
        \end{equation}
        Otherwise, censorship is said to be breakable.
    \end{definition}
    In other words, censorship is secure if $\Omega^{\otimes N}$ acts resource reducing on multipartite states.
    Conversely, if there exists a resource state $\rho\notin \mathcal{F}(A_1\dots A_N)$ which can be restored after censorship, i.e., $\Omega^{\otimes N}(\rho)\rightarrow\rho$, then the censorship is breakable.

    For example, if each sender-receiver pair acts for themselves the protocol simplifies to
    \begin{equation*}
    	\Qcircuit @C=1em @R=.7em {
    		\lstick{\rho_{A_1}} & \qw &\gate{\Omega} & \qw & \gate{\Lambda_{B_1}} & \qw & B_1\\
    		\lstick{\vdots} & &\vdots & & \vdots & &\vdots\\
    		\lstick{} & & & & & &\\
    		\lstick{\rho_{A_N}} & \qw &\gate{\Omega} & \qw & \gate{\Lambda_{B_N}} & \qw & B_N.\\
    	}
    \end{equation*}
    Here, each sender $A_a$ transmits their state $\rho_{A_a}$ individually, resulting in a product state $\rho_{A_1}\otimes\dots\otimes\rho_{A_N}$.
    Subsequently, the receivers obtain the censored state $\Omega(\rho_{A_1})\otimes\dots\otimes \Omega(\rho_{A_N})$, on which each of them applies a free operation $\Lambda_{B_a}\in\mathcal{O}(B_a)$ separately. 
    The final state
    \begin{equation}
        \Lambda_{B_1}(\Omega(\rho_{A_1}))\otimes \dots \otimes \Lambda_{B_N}(\Omega(\rho_{A_N})),
    \end{equation}
    cannot coincide with the original state, $\rho_{A_1}\otimes\dots\otimes \rho_{A_N}$, because of $\Omega(\rho_{A_a})\nrightarrow \rho_{A_a}$, see (i).
    The users' attempt to break the censorship failed and the protocol is secure. 
    More formally, we have the following theorem.
    \begin{theorem}
        \normalfont
        \label{th:secure-vs-local}
        Let $(\mathcal{F},\mathcal{O})$ be a QRT with  
        \begin{equation*}
            \begin{split}
                \mathcal{F}(A_1\dots A_N)&=\mathcal{F}(A_1)\otimes\dots\otimes\mathcal{F}(A_N),\\
                \mathcal{O}(B_1\dots B_N)&=\mathcal{O}(B_1)\otimes\dots\otimes\mathcal{O}(B_N).
            \end{split}
        \end{equation*}
        Then censorship is secure.\hfill$\blacksquare$
    \end{theorem} 

    Theorem~\ref{th:secure-vs-local} assumes the QRT to have a tensor-product structure; see Eq.~\eqref{eq:non-active}.
    In contrast, consider a resource that can be activated, i.e., there is $\sigma\in\mathcal{F}(A)$ such that $\sigma^{\otimes N}\notin \mathcal{F}(A_1\dots A_N)$.
    Despite being a resource, $\sigma^{\otimes N}$ passes the network unharmed, i.e., $\Omega(\sigma)^{\otimes N}=\sigma^{\otimes N}$ via (ii).
    The censorship is breakable.
    Such resources include Bell nonlocality \cite{P12}, not fully-separable states \cite{PV22}, and magic states \cite{BK05,AC12}.
    Simply speaking, any resource that can be distilled by having enough copies of free states cannot be censored in a decentralized network.

    Finally, we note that in a perfect censorship, $\Omega$ leaves no quantum resources between two users.
    That is, it is a resource-destroying channel \cite{LH17}
    \begin{equation}
        \forall \rho\notin\mathcal{F}(A):\quad \Omega(\rho)\in\mathcal{F}(B).
    \end{equation}
    Then the above censorship coincides with the protocols discussed in Ref. \cite{PS24}.
    This idealized variant of censorship may not be achievable in practice, motivating the study of secure censorship based on the (in)convertibility of quantum states, which is the subject of this paper.

    \subsection{Secure censorship and classes of free operations}
    \label{ssec:op-class}

    Theorem \ref{th:secure-vs-local} shows that the users need to cooperate to break the censorship.
    Determining whether the censorship remains secure when users start to team up depends heavily on the joint free operations $\mathcal{O}(B_1\dots B_N)$ available to them.
    The most basic attempt to break censorship is for the receivers to use all (joint) operations that satisfy Eq.~\eqref{eq:res-non-gen}.
    In many concrete situations, the free operations are further restricted.
    For example, if the receivers are spatially well-separated they may only have access to local operations and classical communication \cite{CL14}.
    
    Moreover, according to Def.~\ref{def:secure}, the receivers have to perfectly restore the input state to break the censorship.
    One may relax this condition to obtain a more realistic quantum game \cite{EW00}.
    For instance, one allows for an $\epsilon$ neighborhood $\mathcal{B}_{\epsilon}(\rho)$ around the state $\rho$ and deems the censorship breakable if the final state, $(\Lambda\circ \Omega^{\otimes N})(\rho)$, lies within that neighborhood, for some free operation $\Lambda\in\mathcal{O}(B_1\dots B_N)$. 
    More formally, one writes 
    \begin{equation}
        \Omega^{\otimes N}(\rho)\rightarrow_{\epsilon} \rho~\Leftrightarrow~ \exists\,\omega\in\mathcal{B}_\epsilon(\rho):~\Omega^{\otimes N}(\rho)\rightarrow\omega.
    \end{equation}
    For specific purposes, the set of states $\mathcal{B}_\epsilon(\rho)$ can be defined using a geometric distance, a close-to-unit fidelity, or weak statistical mixtures of $\rho$. 
    
    As an application, one imagines the receivers restoring the original state using only a predefined set $\{\Phi_a\}_a$ of cheap gadgets.
    A sequence $\Lambda=\prod_{a=1}^{n(\epsilon)}\Phi_a$ of these is considered a free operation if the number of gadgets $n(\epsilon)$ has a mild, i.e., poly-logarithmic, dependency on $\epsilon$.
    In this case, censorship is breakable if the pre-censorship state can be restored efficiently.

    Alternatively, we can also relax the condition that breaking the censorship has to be achieved deterministically \cite{CG19}. 
    Instead, the receivers may implement a free operation $\Lambda = \sum_a \Lambda_a,$
    such that 
    \begin{equation}
        \big(\Lambda_a\circ\Omega^{\otimes N}\big)(\rho)=p_a\rho,,
    \end{equation}
    for some outcome $a$ occurring with probability 
    \begin{equation}
        p_a=\Tr\{(\Lambda_a\circ\Omega^{\otimes N})(\rho)\}.
    \end{equation}
    One then restores $\rho$ with probability $p_a$ by performing $\Lambda$ and then measuring $a$ \cite{PSM25}.
    If $\rho$ is restored from the censored state $\Omega^{\otimes N}(\rho)$ by a probabilistic transformation $\Lambda_a$, we write $\Omega^{\otimes N}(\rho)\rightarrow_{p} \rho$.

    Establishing secure censorship against all probabilistic transformations is a very demanding task for the agency (the network provider), because the users need only to transform $\Omega^{\otimes N}(\rho)$ into $\rho$ with some nonzero probability, regardless of how small this probability may be. 
    A more reasonable variant of this protocol will consider the convertibility relation
    \begin{equation}
        \Omega^{\otimes N}(\rho)\rightarrow_{p_{\mathrm{sec}}} \rho~\Leftrightarrow~ \exists\, p\geq p_{\mathrm{sec}}:~\Omega^{\otimes N}(\rho)\rightarrow_{p} \rho.
    \end{equation}
    That is, censorship is breakable, if
    $\Omega^{\otimes N}(\rho)$ can be transformed into $\rho$ with at least probability $p_{\mathrm{sec}}$.
    The security standard set by $p_{\mathrm{sec}}$ may vary on the number of users $N$ trying to break the censorship.
    
    In summary, whether a censorship is secure depends on the capability of the users in a network to convert quantum resources.
    The more we relax the conditions for successful state transformation by free operations, the more difficult it is to restrict the users communication to free states only, i.e., to maintain secure censorship. 
    Which specific class of free operations and convertibility should be considered, will vary drastically on what one agency deems desirable for their network.
    While, our framework (Def. \ref{def:secure}) has the sufficient generality to cover all cases of interest, the next section focuses on censorship against catalytic transformations.

    \section{Can account sharing defy censorship}
    \label{sec:account}

    Consider a quantum network with additional parties $C_1,\dots,C_M$ which are not censored.
    This could be due to these parties having premium access to unlimited quantum communication on the network, or if they can communicate via black-market quantum channels.
    If these independents team up with the censored users, the information-processing protocol for censorship becomes
    \begin{equation*}
    	\Qcircuit @C=1em @R=.7em {
            \lstick{A_1} & \qw & \gate{\Omega} & \qw & \multigate{7}{~\Lambda~} & \qw & \rstick{B_1}\qw\\
    		\lstick{\vdots} & &\vdots & &  &  & \qquad\vdots\\
    		\lstick{} & & & & & \\
    		\lstick{A_N} & \qw & \gate{\Omega} & \qw & \ghost{~\Lambda~} & \qw & \rstick{B_N}\qw\\
    		\lstick{C_1} & \qw & \qw & \qw & \ghost{~\Lambda~} & \qw & \gate{\Tr} \\
    		\lstick{\vdots} & &\vdots & & & & \vdots\\
    		\lstick{} & & & & & \\
    		\lstick{C_M} & \qw & \qw & \qw & \ghost{~\Lambda~} & \qw & \gate{\Tr},\\
    	}
    \end{equation*}
    where the independent parties are traced out at the end of a communication.
    An attempt to break censorship starts with the preparation of a joint resource state
    \begin{equation}
        \label{eq:state-user-indep}
        \rho\otimes\omega_C,~~\text{with}~~\rho\notin\mathcal{F}(A_1\dots A_N)
    \end{equation}
    belonging to the users and $\omega_C$ is prepared by the independents.
    The censorship transforms the input state into $\Omega^{\otimes N}(\rho)\otimes\omega_C$.
    Subsequently, the receivers $B_1,\dots, B_N$ apply a joint free operation $\Lambda\in\mathcal{O}(B_1\dots C_M)$ to restore $\rho$.
    The final state reads
    \begin{equation}
        \label{eq:cat-trafo}
        \Gamma\big( \Omega^{\otimes N}(\rho)\big)=\Tr_{C_1\dots C_M}\big\{\Lambda\big(\Omega^{\otimes N}(\rho)\otimes\omega_C\big)\big\}.
    \end{equation}
    The role of the independents can be understood as providing a catalyst $\omega_C$ for the simulation of the quantum channel $\Gamma$ \cite{LW24}.
    We denote the set of catalytic transformations that can be achieved in this way by $\mathcal{CO}(B_1\dots B_N)$.
    While a free operation in $\mathcal{O}(B_1\dots B_N)$ cannot create a resource [Eq.~\eqref{eq:res-non-gen}], a catalytic transformation $\Gamma\in\mathcal{CO}(B_1\dots B_N)$ might do so.
    If a state $\rho$ can be converted into a state $\gamma$ using catalysis we write $\rho\cra\gamma$.
    If the catalyst $\omega_C$ is used up during the protocol, this is referred to as correlated catalysis \cite{LW24} to distinguish it from strict catalysis \cite{JP99}, in which the catalyst is given back to the independent parties after the protocol.

    In Eq.~\eqref{eq:state-user-indep} we could have allowed for a slightly more general case of catalysis, where the input state is entangled between the users and the independents.
    While this is a valid communication scenario, the resulting catalytic transformation does not need to be completely positive \cite{DL16}.
    This can be witnessed in the time-evolution of initially correlated open quantum systems \cite{P94}.

    \subsection{Censorship against account sharing}
    
    Censorship is secure against account sharing if 
    \begin{equation}
        \label{eq:recall-secure}
        \forall \rho\notin \mathcal{F}(A_1\dots A_N):\quad  \Omega^{\otimes N}(\rho)\nrightarrow_{\mathrm{c}}\rho.
    \end{equation}

    As one would expect, if censorship is secure against account sharing by independent parties, this automatically implies that censorship is secure against the users coordinating with free operations only. 
    Formally, the free operations are a subset of the catalytic transformations,
    \begin{equation}
        \label{eq:cat-inclu}
        \mathcal{O}(B_1\dots B_N)\subseteq \mathcal{CO}(B_1\dots B_N).
    \end{equation}
    For a network without independent parties the inclusion in Eq.~\eqref{eq:cat-inclu} becomes an equality.
    The following discussion then reduces to the argument in Sec. \ref{sec:censor}.

    Catalytic transformations \eqref{eq:cat-trafo} can have the potential to create a resource, i.e., $\Gamma(\sigma)\notin \mathcal{F}(B_1\dots B_N)$, for some free state $\sigma\in\mathcal{F}(A_1\dots A_N)$.
    In this case, censorship is breakable, because if the receivers are able to create a quantum resource, then they can also restore a resource state $\rho\notin\mathcal{F}(A_1\dots A_N)$.
    To see this formally, let $\rho=\Gamma(\sigma)$ be the resource state created by catalysis.
    We concatenate the catalytic transformation $\Gamma$ with the free operation $\Lambda(\rho)=\Tr(\rho)\sigma$ [cf. Eq. \eqref{eq:prep-free-state}].
    This yields
    \begin{equation}
        \Gamma\big(\Lambda\big[\Omega^{\otimes N}(\rho)\big]\big)=\Gamma(\sigma)=\rho.
    \end{equation}
    Thus, $\rho$ has been restored by $\Gamma\circ \Lambda$.
    Thus, $\Omega^{\otimes N}(\rho)\cra\rho$ and censorship is breakable.
    In summary, censorship is secure only if catalytic transformations cannot create a resource.

    \subsection{Swapping quantum states breaks censorship}

    An example that illustrates the above situation well is the swap operation, $\Phi_{\mathrm{S}}(\rho\otimes \sigma)=\sigma\otimes \rho$, which exchanges the quantum states of two parties.
    On its own, $\Phi_{\mathrm{S}}$ is a rather benign operation as its does not create a resource: for $\sigma\in\mathcal{F}(A)$ and $\gamma\in\mathcal{F}(A^\prime)$, we have that
    \begin{equation}
        \label{eq:swap-no-res}
        \Phi_{\mathrm{S}}(\sigma\otimes\gamma)=\gamma\otimes \sigma\in \mathcal{F}(B)\otimes \mathcal{F}(B^\prime),
    \end{equation}
    is again a free state.
    However, the situation changes as soon as uncensored parties join the quantum communication.
    The users coordinate with the independents to realize the catalytic transformation
    \begin{equation}
        \Gamma_{\omega}(\rho)=\Tr_{C_1\dots C_M}\big\{\Phi_{\mathrm{S}}\big(\rho\otimes\omega\big)\big\}=\omega,
    \end{equation}
    where $\Phi_{\mathrm{S}}$ swaps the state $\rho$ of the users with the state $\omega$ of the independents.
    If the catalyst is chosen to be the desired output state, i.e., $\omega=\rho$, then after censorship the receivers obtain 
    \begin{equation}
        \Gamma_\rho\big(\Omega^{\otimes N}(\rho)\big)=\rho.
    \end{equation}
    Clearly, $\Omega^{\otimes N}(\rho)\cra \rho$ was achieved, and censorship is therefore breakable.

    \begin{theorem}
        \normalfont
        \label{th:swap}
        Let $\Phi_{\mathrm{S}}\in\mathcal{O}(B_1 \dots C_M)$ be a free operation.
        Then censorship is breakable by account sharing.\hfill$\blacksquare$
    \end{theorem} 

    Swapping or exchanging quantum states between two or multiple parties in a network presents an interesting special case within a QRT.
    On one hand, these permutations do not introduce quantum resources on their own; see Eq. \eqref{eq:swap-no-res}.
    On the other hand, permutations allow for the undisturbed exchange of quantum states, thus providing an uncensored communication channel.
    The swap operation can therefore be understood as a free but not completely-free operation \cite{CV20}.
    Indeed, the swap operation is not always considered to be available due to its nonlocal nature.
    For example, it cannot be realized by local operations and classical communication \cite{HN03}.

    \subsection{Secure censorship is well defined}

    Free operations order quantum states according to their resourcefulness \cite{SV15,CG19}.
    It is a generic fact of such an ordering that for certain states it is not possible to say that one is more resourceful than the other.
    That is, there are states $\rho$ and $\omega$ for which
    \begin{equation}
        \label{eq:non-comp}
        \rho\nrightarrow\omega \quad \text{and} \quad \omega \nrightarrow \rho.
    \end{equation}
    For example, in the QRT of entanglement, certain tripartite states such as the GHZ and W state cannot be converted into each other by (stochastic) local operations and classical communication \cite{BPR00}.  
    Likewise, in the QRT of coherence, there are states for which neither one of them can be transformed into the other by applying random unitaries \cite{N99}.
    
    In resource censorship, the output state of the network $(\Lambda\circ\Omega^{\otimes N})(\rho)$ is always comparable to the input state $\rho$.  
    This follows, because $\Omega^{\otimes N}$ is itself a free operation, which can be seen from $\Omega\in\mathcal{O}(A\to B)$ (Def. \ref{def:RR}) and the fact that the underlying QRT admits a tensor-product structure, see Eq. \eqref{eq:para-free-ops}.
    Thus, $\rho\rightarrow \Omega^{\otimes N}(\rho)$.
    In summary, censorship is either secure or breakable, and ambiguous cases as in Eq. \eqref{eq:non-comp} cannot occur (Fig. \ref{fig:regions}). 

    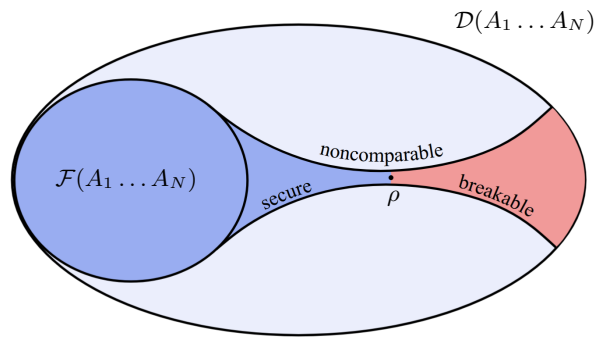
\begin{figure}[t]
        \centering
        \begin{tikzpicture}
        \node at (0,0) {\includegraphics[width=0.45\textwidth]{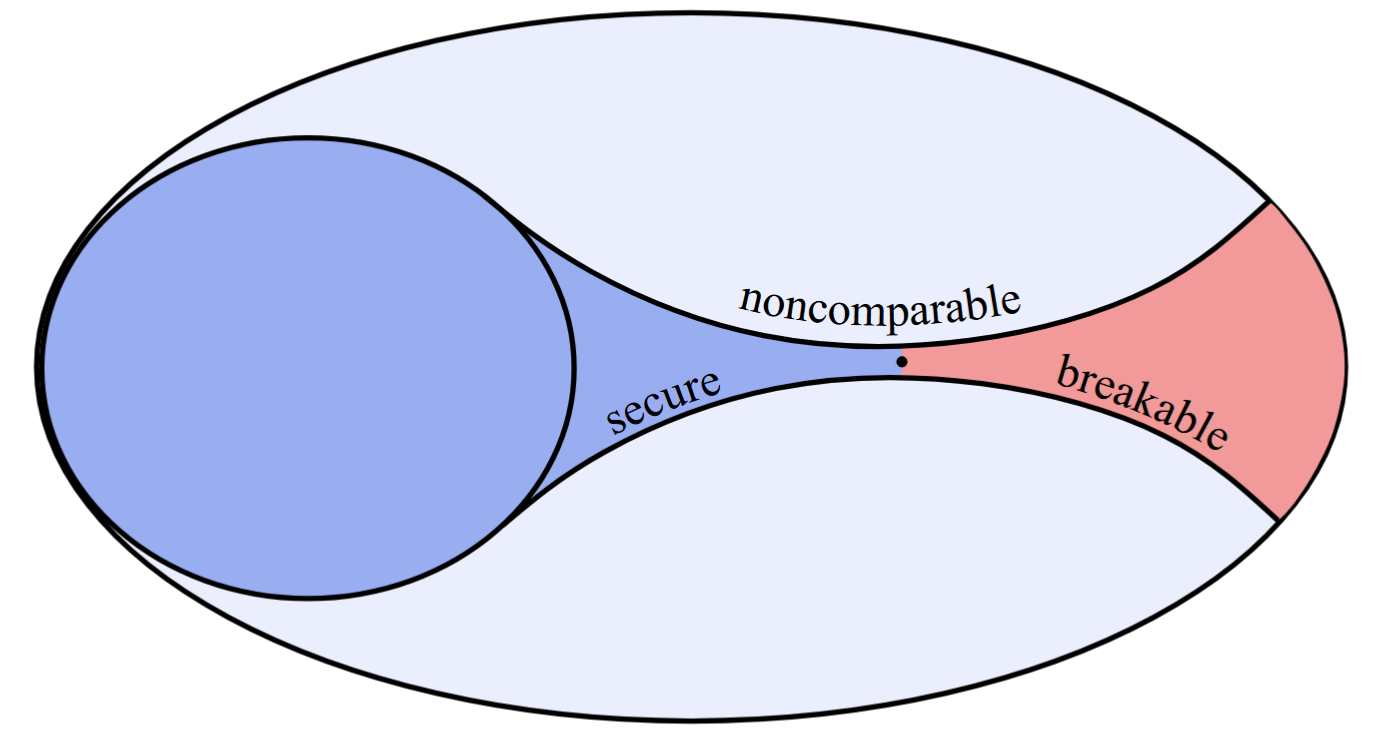}};
        \node at (3,2.1) {$\mathcal{D}(A_1\dots A_N)$};
        \node at (-2.3,0) {$\mathcal{F}(A_1\dots A_N)$};
        \node at (1.25,-0.25) {$\rho$};
        \end{tikzpicture}
        \caption{\label{fig:regions}
        Ordering of multipartite states by free operations.
        If the state after censorship $\Omega^{\otimes N}(\rho)$ lies in the blue region it is less resourceful than $\rho$, i.e., $\Omega^{\otimes N}(\rho)\nrightarrow \rho$. 
        States in the red region allow for a recovery of $\rho$ using free operations, marking a breakable censorship.
        In the light gray region lie states which cannot be freely converted into $\rho$, but neither can $\rho$ be freely converted into them.
        The state after censorship $\Omega^{\otimes N}(\rho)$ cannot lie in these regions, as $\Omega^{\otimes N}$ is itself a free operation.
        }
    \end{figure}

    \subsection{On the existence of censorship protocols}

    So far, we considered a resource-reducing channel to be given without commenting on the assumptions this entails.
    For a resource-reducing channel $\Omega$ to exist, the set $\mathcal{F}$ must be affine.
    That is, for free states $\sigma_a\in\mathcal{F}$, the state $\sigma=\sum_a t_a \sigma_a$, with $t_a\in\mathbb{R}$, is again free.
    Due to (ii) in Def. \ref{def:RR} we have $\Omega(\sigma_a)=\sigma_a$.
    By linearity it follows that
    \begin{equation}
        \Omega(\sigma)=\sum_a t_a \Omega(\sigma_a)=\sigma
    \end{equation}
    must be free as well.
    In other words, resource-reducing channels for nonaffine QRTs are nonlinear \cite{LH17,PM24}.
    QRTs that are affine include coherence and asymmetry for which censorship is devised in Sec.~\ref{sec:coherence} and~\ref{sec:asymmetry}, respectively.
    
    A QRT is convex if the coefficients $t_a\geq 0$ are probabilities.
    While it is clear that any affine QRT is convex, the converse is not true, with the prime example being entanglement \cite{HH09}.
    There, the free states $\mathcal{F}(A)$ of a $K$-partite system $A$ consists of separable states, $\sigma = \sum_a p_a \rho^a_1\otimes\dots\otimes \rho^a_n$, which are statistical mixtures of product states, $p_a\geq 0$.
    The set $\mathcal{F}(A)$ is convex but not affine because any state $\rho\in\mathcal{D}(A)$ can be written as a sum over product states once we allow for negative $t_a\in\mathbb{R}$ \cite{CG19,PM25}.
    Thus, no entanglement-reducing channel $\Omega$ exists, and this raises concerns about the limited applicability of censorship protocols.
    
    To circumvent this limitation, we can divide a convex set $\mathcal{F}$ into affine subsets $\mathcal{F}_k$ such that $\mathcal{F}=\bigcup_k \mathcal{F}_k$. 
    A sender $A$ in a quantum network then has to append a classical key $k$ to their quantum message $\sigma$.
    The key $k$ specifies the subset $\mathcal{F}_k$ in which $\sigma$ lies.
    Based on $k$, the agency applies a resource-reducing channel $\Omega_k$ before allowing the state $\sigma$ to be passed to a receiver $B$. 
    Free states $\sigma\in\mathcal{F}_k$ are left unchanged, $\Omega_k(\sigma)=\sigma$, 
    The channel $\Omega_k$ will degrade the resources of any state $\rho\notin\mathcal{F}$, thus enforcing censorship on the communication. 
    Security of the censorship must be independent of the specific key $k$, because the agency has to treat each user $A$ as an untrusted party, i.e., $\Omega_k(\rho)\nrightarrow \rho$ for all $k$.

    The above construction is always possible as one can choose $\mathcal{F}_k=\{\sigma_k\}$ to contain a single free state $\sigma_k\in\mathcal{F}$.
    In this edge case, we have $\Omega_k(\rho)=\Tr(\rho)\sigma_k$.
    Conditional censorship was devised in Ref.~\cite{PM24} for an ideal censorship, that is, $\Omega_k$ is a resource-destroying channel, but is readily extended to imperfect censorship.

    \section{Quantum coherence}
    \label{sec:coherence}

    Superpositions are the most fundamental resource that differentiates classical and quantum physics.
    In the QRT of coherence \cite{BC14,LM14,SW18}, one quantifies the amount of superpositions in a general mixed state with respect to a fixed orthonormal basis $\{\ket{a}\}_{a=0}^{d-1}$, the incoherent basis. 
    Free (likewise, incoherent) states are diagonal in that basis, 
    \begin{equation}
        \mathcal{F}(A)=\Big\{\sigma\in\mathcal{D}(A)\,\Big|\, \sigma=\sum_{a}p_a\ket{a}\bra{a},\, p_a\geq 0\Big\}.
    \end{equation}
    We consider the free operations of this QRT to consist of the incoherent operations $\mathcal{O}_{\mathrm{IO}}(A\to B)$ \cite{BC14,SV15},
    \begin{equation}
        \label{eq:IO-Kraus}
        \Lambda(\rho)=\sum_a M_a \rho M_a^\dag,\quad M_a=\sum_{b=0}^{d-1}c_{ab}\ket{f(b)}\bra{b},
    \end{equation}
    with $f(b)\in\{0,\dots, d-1\}$. 
    Having specified a QRT, we are now in a position to give a coherence-reducing channel.
    
    \begin{proposition}
        \normalfont
        \label{prop:RR-coherence}
        A coherence-reducing channel is
        \begin{equation*}
            \label{eq:resource-reducing-coherence}
            \Omega(\rho)=\epsilon\rho + (1-\epsilon)\Delta(\rho),
        \end{equation*}
        where $\epsilon\in[0,1)$ and $\Delta(\rho)=\sum_a \ket{a}\bra{a}\rho\ket{a}\bra{a}$.
    \end{proposition}
    \begin{proof}
        First, $\Omega\in\mathcal{O}_{\mathrm{IO}}(A\to B)$, because it is a mixture of the two incoherent operations $\id,\Delta\in\mathcal{O}_{\mathrm{IO}}(A\to B)$.
        To show that $\Omega$ satisfies (i) in Def.~\ref{def:RR} suppose the converse: there is $\Lambda\in\mathcal{O}_{\mathrm{IO}}(B)$ such that $\rho$ is restored, i.e.,
        \begin{equation}
            (\Lambda\circ\Omega)(\rho)=\epsilon\Lambda(\rho) + (1-\epsilon)(\Lambda\circ\Delta)(\rho)=\rho.
        \end{equation}
        This would imply for the off-diagonal elements
        \begin{equation}
            \label{eq:contra}
            \forall a\neq b:\quad \epsilon\Lambda(\rho)_{ab}=\rho_{ab},
        \end{equation}
        where we used that $\Lambda(\Delta(\rho))_{ab}=0$, for $a\neq b$, which follows from $(\Lambda\circ\Delta)(\rho)\in\mathcal{F}(B)$ being an incoherent state.
        Since $\epsilon<1$, Eq. \eqref{eq:contra} implies that $\Lambda$ increases the off-diagonals by a factor of $\epsilon^{-1}$.
        This is not possible using channels of the form in Eq.~\eqref{eq:IO-Kraus}.
        This leaves us with a contradiction, showing that $\Omega$ satisfies (i).
        Finally, $\Omega$ leaves incoherent states $\sigma\in\mathcal{F}$ unaltered, $\Omega(\sigma)=\sigma$, because of $\Delta(\sigma)=\sigma$.
        Thus (ii) in Def.~\ref{def:RR} holds as well, and $\Omega$ is coherence-reducing.
    \end{proof}

    The coherence-reducing channel $\Omega$ creates a statistical mixture of $\rho$ with the incoherent state $\Delta(\rho)\in\mathcal{F}(B)$.  
    The channel $\Omega$ is encountered when modeling dephasing.
    It emerges from the (Markovian) interaction of an open quantum system with the environment resulting in a monitoring of the incoherent states $\ket{a}$.

    \subsection{Censorship of quantum coherence}

    The QRT of coherence extends naturally to joint systems.
    Here, $\mathcal{F}(A_1\dots A_N)$ consists of statistical mixtures
    \begin{equation}
        \label{eq:multi-incoh}
        \rho=\sum_{a_1,\dots, a_N} p_{a_1\dots a_N} \ket{a_1\dots a_N}\bra{a_1\dots a_N},
    \end{equation}
    of product states $\ket{a_1\dots a_N} = \ket{a_1}\otimes \dots \otimes \ket{a_N}$, with $p_{a_1\dots a_N}$ being a probability distribution.
    The $N$-partite incoherent operations $\Lambda\in\mathcal{O}_{\mathrm{IO}}(B_1\dots B_N)$ are as in Eq. \eqref{eq:IO-Kraus} but with $f(b_1,\dots, b_N)$ taking values in the $N$-fold Cartesian product $\{0,\dots,d-1\}^{\times N}$.
    
    Imposing a censorship on coherence using $\Omega$ means that only incoherent states are preserved during the communication, while states containing superpositions are suppressed.
    Without any independents in the network, the censorship of coherence turns out to be secure.
    \begin{theorem}
        \normalfont
        \label{th:coh-secure}
        Censorship of coherence $(\mathcal{F},\mathcal{O}_{\mathrm{IO}})$ is secure.
    \end{theorem}
    \begin{proof}
        For $\rho\notin\mathcal{F}(A_1\dots A_N)$, the censorship yields
        \begin{equation}
            \label{eq:RR-prod}
            \begin{split}
                \Omega^{\otimes N}(\rho)=\epsilon^N \rho + \dots + (1-\epsilon)^N\Delta^{\otimes N}(\rho).
            \end{split}
        \end{equation}
        Arguments similar to the ones in Proposition \ref{prop:RR-coherence} show that there is no $\Lambda\in\mathcal{O}(B_1\dots B_N)$ such that $\Lambda\big(\Omega^{\otimes N}(\rho)\big)=\rho$.
        Thus, $\Omega^{\otimes N}(\rho)\nrightarrow \rho$ and censorship is secure.
    \end{proof}
    Secure censorship is a positive result for any agency trying to reserve quantum communication for specific costumers, while restricting the general users to classical communication and shared randomness; see Eq. \eqref{eq:multi-incoh}.

    Note that, the notion of secure censorship is independent of the specific value $\epsilon$ in Proposition \ref{prop:RR-coherence}.    
    Nevertheless, in real-world applications, the agency may want to achieve a sizable reduction of coherence before a state reaches a receiver.
    In an ideal censorship, we have $\epsilon=0$,
    so that no coherences are left in the output state, i.e.,
    \begin{equation}
        \Omega(\rho)=\Delta(\rho)\in\mathcal{F}(B),
    \end{equation}
    is an incoherent state.
    Only in this case does $\Omega$ also break the entanglement to any adjacent user \cite{HS03}.
    In practice, $\epsilon=0$, may not be achieved perfectly.

    \subsection{Revival of quantum coherence by account sharing}

    The situation becomes more intricate when there are independent parties $C_1,\dots,C_M$ in the network, which aid the transmission of quantum coherence by the users.
    In particular, the set $\mathcal{O}_{\mathrm{IO}}(B_1\dots C_M)$ contains the swap operation $\Phi_{\mathrm{S}}$, because it has the form \eqref{eq:IO-Kraus}, with $M=\sum_{a,b}\ket{ab}\bra{ba}$.
    The swapping of a users' state with an independent party breaks the censorship. 
    \begin{theorem}
        \label{th:coh-break}
        \normalfont
        The censorship of coherence $(\mathcal{F},\mathcal{CO}_{\mathrm{IO}})$ is breakable by account sharing.
    \end{theorem}
    \begin{proof}
        The result follows from Theorem \ref{th:swap}.
    \end{proof}

    Despite the swap channel being an incoherent operation, it is a very powerful one.
    Exchanging arbitrary quantum states across long distances means the receivers and the independents are effectively connected by an uncensored communication channel.
    More often we may be concerned with a situation in which receivers use classical communication to coordinate local operations with each other and the independents.
    That is, they can apply an incoherent separable operation \cite{SRB17},
    \begin{equation}
        \label{eq:sep-map}
        \Lambda_{B_1\dots C_M}=\sum_{b} \Lambda_{B_1}^b\otimes\dots \otimes \Lambda_{B_N}^b \otimes \Lambda_{C_1}^b\otimes \dots\otimes \Lambda_{C_M}^b,
    \end{equation}
    with $\Lambda_{B_a}^b$ being a (possibly probabilistic) operation of the form in Eq.~\eqref{eq:IO-Kraus}.
    We denote the set of operations in Eq.~\eqref{eq:sep-map} by
    $\mathcal{O}_{\mathrm{SIO}}(B_1\dots C_M)$.
    This set is strictly smaller than the set of incoherent operations $\mathcal{O}_{\mathrm{IO}}(B_1\dots B_M)$.
    In particular, it does not contain the swap operation,
    $\Phi_{\mathrm{S}}\notin\mathcal{O}_{\mathrm{SIO}}(B_1\dots B_M)$.
    For these more restricted operations the censorship is again secure.
    \begin{theorem}
        \normalfont
        \label{th:coh-secure-2}
        The censorship of coherence $(\mathcal{F},\mathcal{CO}_{\mathrm{SIO}})$ is secure against account sharing.
    \end{theorem}
    \begin{proof}
        For any resource state $\rho\otimes\omega$ between the users and the independents and every free operation $\Lambda\in\mathcal{O}_{\mathrm{SIO}}(B_1\dots C_M)$, the (reduced) state of $B_1,\dots, B_N$ is 
        \begin{equation}
            \begin{split}
                &\Tr_{C_{1}\dots C_M}\big\{\Lambda\big(\Omega^{\otimes N}(\rho)\otimes \omega\big)\big\}\\ 
                &=\sum_{b}\big[\big(\Lambda_{B_1}^b\otimes \dots \otimes\Lambda_{B_N}^b\big)\circ \Omega^{\otimes N}\big](\rho)\\
                &\quad \times\Tr\big\{ \big(\Lambda_{C_{1}}^b\otimes \dots \otimes\Lambda_{C_M}^b\big)(\omega)\big\}=\Lambda^\prime \big[\Omega^{\otimes N}(\rho)\big],
            \end{split}
        \end{equation}
        where we derived the catalytic transformation
        \begin{equation}
            \Lambda^\prime=\sum_{b}\lambda_{\omega}^b \Lambda_{B_1}^b\otimes\dots \otimes \Lambda_{B_N}^b,
        \end{equation}
        having nonnegative coefficients 
        \begin{equation}
            \lambda_{\omega}^b=\Tr\big\{ \big(\Lambda_{C_{1}}^b\otimes \dots \otimes\Lambda_{C_M}^b\big)(\omega)\big\}\geq 0.
        \end{equation}
        It is not hard to see that $\Lambda^\prime\in\mathcal{O}_{\mathrm{SIO}}(B_1\dots B_N)$ is a separable operation.
        Theorem \ref{th:coh-secure} implies that $\rho$ cannot be restored using incoherent operations, and thus it cannot be restored using $\Lambda^\prime\in\mathcal{O}_{\mathrm{SIO}}(B_1\dots B_N)$.
        Thus, $\Omega^{\otimes N}(\rho)\nrightarrow_{\mathrm{c}} \rho$, for any $\rho\notin\mathcal{F}(A_1\dots A_N)$.
    \end{proof}
    
    Theorem \ref{th:coh-secure-2} establishes that the independents providing a catalyst does not suffice to break the censorship, because the receivers $B_1,\dots, B_N$ lack the ability to transfer coherences from the catalyst.
    Mathematically, we have $\mathcal{CO}_{\mathrm{SIO}}=\mathcal{O}_{\mathrm{SIO}}$, that is, catalysis does not enlarge the set of free operations.
    In particular, it does not allow to simulate a non-free operation.

    When removing classical communication as well, the users and independents are only connected via local incoherent operations and shared randomness.
    This captures the restrictions found in Bell-like experiments and nonlocal games \cite{B12}.
    Since such situations are a special case of separable incoherent operations Theorem \ref{th:coh-secure} ensures that in these games no transfer of coherence is possible between the players.

    \section{Quantum asymmetry}
    \label{sec:asymmetry}

    Certain types of quantum information are, without a common reference, of no use to the communicating parties in a network. 
    For example, if two users want to communicate via the exchange of a spin-1/2 particle, they may first agree on a spatial orientation so as to measure spin components relative to a Cartesian frame \cite{BRS03}.
    Likewise, if a sender's laser acts as a phase reference, but is not phase locked to the one of a receiver, then this results in a loss of coherence during transmission \cite{M97,BR07}.
    Mathematically, a reference frame is given by an element $g$ of a (here, finite) group $\mathcal{G}$.
    Each frame $g$ is represented by a unitary matrix $U^g$ such that the state $\rho$ of a sender $A$ is described by $U^g\rho U^{g\dag}$ in the frame of a receiver $B$.
    If $B$ lacks knowledge of $g$, then $B$'s state is symmetric under the action of any $U^g$, i.e., $U^g\sigma U^{g\dag} =\sigma$ for all $g\in\mathcal{G}$. 
    For illustration, rotating a spherically symmetric object toward an agreed-upon direction has no significance.
    
    In the QRT of asymmetry \cite{JB03,VAW08}, the lack of a shared reference defines the free (likewise, symmetric) states, 
    \begin{equation}
        \mathcal{F}(A) = \big\{\sigma\in\mathcal{D}(A)\,\big|\,\forall g:\, U^g\sigma U^{g\dag} = \sigma\big\},
    \end{equation}
    which are obtained by averaging over all possible $g$, i.e.,
    \begin{equation}
        \label{eq:RD-asym}
        \Delta(\rho)=\frac{1}{|\mathcal{G}|}\sum_{g\in\mathcal{G}} U^g \rho U^{g\dag},
    \end{equation}
    known as the twirling channel.
    It holds that $\Delta(\sigma)=\sigma$ for all $\sigma\in\mathcal{F}(A)$.
    The free operations of this QRT consist of the covariant operations $\mathcal{O}_{\mathrm{Cov}}(A\to B)$ \cite{BR07,CG19},
    \begin{equation}
        \forall g:~\Lambda(U^g \rho U^{g\dag}) = U^g \Lambda(\rho) U^{g\dag}.
    \end{equation}
    Having specified a QRT, we can give an asymmetry-reducing channel.

    \begin{proposition}
        \normalfont
        \label{prop:RR-asymmetry}
        An asymmetry-reducing channel is
        \begin{equation*}
            \label{eq:resource-reducing-asymmetry}
            \Omega(\rho)=\sum_{g} p_g U^g \rho U^{g\dag},
        \end{equation*}
        where $0<p_g<1$ are strictly positive probabilities.
    \end{proposition}
    \begin{proof}
        First, $\Omega\in\mathcal{O}_{\mathrm{Cov}}(A\to B)$ is covariant due to $\Omega(U^g\rho U^{g\dag})=U^g \Omega(\rho)U^{g\dag}$ for all $g$.
        To show (i) in Def.~\ref{def:RR} we first note that $\Omega$ increases the von Neumann entropy, $S(\rho)=-\Tr(\rho\log \rho)$, of any asymmetric state $\rho\notin\mathcal{F}$, viz.
        \begin{equation}
            \label{eq:ent-ineq}
            \begin{split}
                S(\Omega(\rho)) & > \sum_g p_g S(U^g\rho U^{g\dag}),\\
                & = \sum_g p_g S(\rho) = S(\rho),\\
            \end{split}
        \end{equation}
        where the second line holds due to $S(U^g\rho U^{g\dag})=S(\rho)$ and $\sum_g p_g = 1$.
        The first line follows from the entropy being concave, $S(\sum_g p_g U^g\rho U^{g\dag})\geq \sum_g p_g S(U^g\rho U^{g\dag})$.
        The inequality~\eqref{eq:ent-ineq} is strict as equality holds only if $U^g\rho U^{g\dag}=\rho$ for all $g$, which is not possible because $\rho\notin\mathcal{F}$ is, by assumption, asymmetric.
        Introducing the (entropy of) frameness, $f(\rho)=S(\Delta(\rho)) - S(\rho)$ \cite{GMS09} and using Eq.~\eqref{eq:ent-ineq} we find $f(\Omega(\rho))<f(\rho)$.
        Since a covariant operation $\Lambda\in \mathcal{O}_{\mathrm{Cov}}(B)$ cannot increase the frameness \cite{LW24},
        \begin{equation}
            f((\Lambda\circ\Omega)(\rho)) \leq f(\Omega(\rho)) < f(\rho),
        \end{equation}
        it follows that $(\Lambda\circ\Omega)(\rho)\neq\rho$, which amounts to $\Omega(\rho)\nrightarrow \rho$ showing (i).
        Finally, $\Omega$ leaves free states $\sigma\in\mathcal{F}$ unaltered,
        \begin{equation}
            \Omega(\sigma) = \sum_g p_g U^g \sigma U^{g\dag} = \sigma,
        \end{equation}
        which follows from $U^g\sigma U^{g\dag}=\sigma $ for all $g$.
        Thus, (ii) in Def.~\ref{def:RR} holds as well, and $\Omega$ is asymmetry reducing.
    \end{proof}
    When the agency attempts to apply the twirling channel $\Delta$ but fails to achieve a perfect averaging, $p_g\neq \tfrac{1}{|\mathcal{G}|}$, they instead realize the channel $\Omega$.
    While $\Delta$ is the unique asymmetry-destroying channel~\cite{G17}, any strictly positive probability distribution, $0<p_g<1$, defines an asymmetry-reducing channel $\Omega$.

    \subsection{Censorship of quantum asymmetry}

    The QRT of asymmetry extends naturally to joint systems. 
    Here, $\mathcal{F}(A_1\dots A_N)$ consists of symmetric states
    \begin{equation}
        \forall g:~(U^g_{1}\otimes \dots \otimes U^g_{N})\sigma (U^g_{1}\otimes \dots \otimes U^g_{N})^\dag = \sigma,
    \end{equation}
    where each subsystem lacks information of the local reference frame $g$.
    The $N$-partite covariant operations $\Lambda\in\mathcal{O}_{\mathrm{Cov}}(B_1\dots B_N)$ obey
    \begin{equation}
        \label{eq:multi-cov-op}
        \begin{split}
            &\Lambda\big((U^g_{1}\otimes \dots \otimes U^g_{N})\rho(U^g_{1}\otimes \dots \otimes U^g_{N})^\dag\big)\\
            &=(U^g_{1}\otimes \dots \otimes U^g_{N})\Lambda(\rho)(U^g_{1}\otimes \dots \otimes U^g_{N})^\dag\\
        \end{split}
    \end{equation}
    for any frame $g$.

    The users of a network might try to establish a reference frame $g^\prime$ for their quantum communication before transmission, e.g., using classical communication. 
    The asymmetry-reducing channel $\Omega$ will nullify their established frame by randomly applying unitary rotations $U^g$ with probability $p_g$.
    Thus, the user's strategy fails and so does any other attempt using covariant operations~\eqref{eq:multi-cov-op} as is embodied in the following theorem.
    \begin{theorem}
        \normalfont
        \label{th:coh-secure}
        Censorship of asymmetry is secure.
    \end{theorem}
    \begin{proof}
        For $\rho\notin\mathcal{F}(A_1\dots A_N)$, the state after censorship is $\Omega^{\otimes N}(\rho)$.
        Strict concavity of the entropy implies $S(\Omega^{\otimes N}(\rho))>S(\rho)$.
        Thus, $f(\Omega^{\otimes N}(\rho))<f(\rho)$ and due to $f(\Lambda(\rho))\leq f(\rho)$ we have $f((\Lambda\circ\Omega^{\otimes N})(\rho))<f(\rho)$ showing that $(\Lambda\circ\Omega^{\otimes N})(\rho)\neq \rho$ for all $\Lambda\in\mathcal{O}_{\mathrm{Cov}}(B_1\dots B_N)$.
        Thus, $\Omega^{\otimes N}(\rho)\nrightarrow \rho$ and censorship is secure.
    \end{proof}
    
    \subsection{Revival of quantum asymmetry by account sharing}

    The situation changes when independents parties $C_1,\dots,C_M$ aid the transmission of asymmetric states.
    They provide a catalyst for the simulation of non-covariant channels. 
    \begin{theorem}
        \normalfont
        \label{th:asym-secure-2}
        The censorship of asymmetry $(\mathcal{F},\mathcal{CO}_{\mathrm{Cov}})$ is breakable by account sharing.
    \end{theorem}
    \begin{proof}
        For any resource state $\rho\notin\mathcal{F}(A_1\dots A_N)$ of the senders, the independents supply a pure state $\ket{\phi}$ acting as an ideal reference frame, i.e., $\braket{\phi|U^g_C|\phi}=0$ for all $U_C^g=U_{C_1}^g\otimes \dots \otimes U_{C_M}^g \neq \mathbb{1}^{\otimes M}$.  
        After censorship $\Omega^{\otimes N}$, the receivers and independents apply the joint covariant operation \cite{MS13}
        \begin{equation*}
            \begin{split}
                &\Lambda \big(\Omega^{\otimes N}(\rho) \otimes \ket{\phi}\bra{\phi}\big) \\
                &= \sum_{g} \Tr\big[(\mathbb{1}^{\otimes N}\otimes U_C^g\ket{\phi}\bra{\phi}U_C^{g\dag})(\Omega^{\otimes N}(\rho) \otimes \ket{\phi}\bra{\phi})\big]\\
                &\quad \times U_B^g\rho U_B^{g\dag}\otimes U_C^g\ket{\phi}\bra{\phi}U_C^{g\dag},
            \end{split}
        \end{equation*}
        where $U_B^g=U_{B_1}^g\otimes \dots \otimes U_{B_N}^g$.
        Then, 
        \begin{equation}
            \Tr_{C_1\dots C_M}\big\{\Lambda \big(\Omega^{\otimes N}(\rho) \otimes \ket{\phi}\bra{\phi}\big)\big\}=\rho.
        \end{equation}
        Thus, $\Omega^{\otimes N}(\rho)\rightarrow_{\mathrm{c}} \rho$, for any $\rho\notin\mathcal{F}(A_1\dots A_N)$.
    \end{proof}
    Theorem \ref{th:asym-secure-2} shows that not only can censorship be overcome but any asymmetric state $\rho\notin\mathcal{F}(A_1\dots A_N)$ can be restored.
    Notably, the state $\ket{\phi}$ provided by the independents was not consumed during the catalysis.

    The security of censorship depends on the group $\mathcal{G}$.  
    Here, we considered a finite group $\mathcal{G}$ for which catalysis is always possible \cite{MS13}, leading to Theorem~\ref{th:asym-secure-2}.
    The analysis can be extended to continuous (connected Lie) groups for which perfect catalysis with pure states as in Theorem \ref{th:asym-secure-2} is never possible \cite{MS13}.
    Then, the state of the independents must be consumed during catalysis \cite{DHF21,LWW23} for the users to have a chance at breaking the censorship.
    
    \section{Conclusion}
    \label{sec:Fin}

    In this work, we introduced novel protocols for resource censorship, aimed at suppressing the transmission of quantum resources in a public-domain network, while allowing for unrestricted communication with free states. 
    Unlike approaches that demand the complete erasure of quantum resources, in our framework censorship is deemed secure if the receiver cannot recover the lost quantum correlations using only free operations.

    Censorship relies on a resource-reducing channel, applied locally to each communication line in a network. 
    The channel acts as an (imperfect) filter, degrading quantum features while leaving classical ones unaffected. 
    From this perspective, secure censorship can be interpreted as the implementation of completely filtering maps—--processes that filter quantum resources even in the presence of entangled ancillary systems.

    We analyzed conditions under which censorship protocols remain secure, especially in the presence of independent or privileged parties who may assist censored users in bypassing the imposed restrictions. 
    The theory was applied to the example of quantum coherence, where censorship corresponds to the users being only permitted access to classical communication, while genuine quantum communication is suppressed.
    We considered different operational restrictions, demonstrating both secure and breakable censorship regimes.
    Moreover, we considered the censorship of asymmetry in which the lack of a shared reference frame limits the users communication.
    There, we considered finite symmetry groups for which catalysis is always possible, and thus censorship is breakable when independent parties assist the quantum communication.

    Resource censorship might be a fruitful endeavor for future research, as it connects to broader phenomena in quantum information, such as quantum catalysis \cite{JP99,CC10}, resource embezzlement \cite{DH03,TS22}, and entanglement-assisted quantum communication \cite{BS99,B02,HD08}.
    For example, we considered cases in which the users consume the resources provided by independent (uncensored) parties, but one could imagine the independents only borrowing the catalyst $\omega$ to the users.
    In this picture, the users seek the broadcasting of a resource \cite{BB07,MS19,SG24}, i.e., $\Omega^{\otimes N}(\rho)\otimes \omega \to \omega \otimes \omega$. 
    
    Moreover, if the resource-reducing channel $\Omega$ is interpreted as an unwanted source of noise (i.e., a known quantum error), then breaking the censorship corresponds to an entanglement-assisted error correction \cite{LB13,WB09,BDH06}.
    Beyond the foundational insights gained by this framework, censorship protocols may have practical implications in the (far-distant?) future.
    In places where quantum and post-quantum cryptography are not at the state of the art, the surveillance of quantum resources of untrusted parties in a network may be a viable alternative to making ones own communication more secure. 
    
    \acknowledgements
    We gratefully acknowledge financial support from Denmarks Grundforskningsfond (DNRF 139, Hy-Q Center for
    Hybrid Quantum Networks) and the Alexander von Humboldt Foundation (Feodor Lynen Research Fellowship).

	\appendix

    \section{Orbits of resource-reducing channels}
    \label{app:resource-reducing}

    A unitary channel $\Phi(\rho)=U\rho U^\dag$, with $U^\dag U=\mathbb{1}_A$, is called a free unitary if each $\Phi\in\mathcal{O}(A\to B)$ and $\Phi^{-1}\in\mathcal{O}(B\to A)$ are free operations.
    We can use free unitaries to obtain different resource-reducing channels for a QRT. 
    \begin{proposition}
        \normalfont
        \label{prop:orbit}
        Let $\Omega$ be resource-reducing.
        Then,
        \begin{equation*}
            \Omega_\Phi=\Phi^{-1}\circ\Omega\circ\Phi,
        \end{equation*}
        is resource-reducing, for any free unitary $\Phi$.
    \end{proposition}
    \begin{proof}
        $\Omega_\Phi$ is a free operation, due to $\mathcal{O}$ being closed under composition.
        Next, we show that $\Omega_\Phi(\rho)\nrightarrow \rho$, for all $\rho\notin\mathcal{F}$.
        This holds, because of clause (i) in Def. \ref{def:RR}, i.e.,
        \begin{equation}
            \label{eq:A3}
            \Omega[\Phi(\rho)]\nrightarrow \Phi(\rho)\notin\mathcal{F}.
        \end{equation}
        Applying the free unitary $\Phi^{-1}$ to both sides of Eq. \eqref{eq:A3} yields (i), i.e., $\Omega_\Phi(\rho)\nrightarrow \rho$.
        Finally, $\Omega_\Phi$ preserves free states, $\Omega_\Phi(\sigma)=\sigma$, which follows from $\Omega(\Phi(\sigma))=\Phi(\sigma)\in\mathcal{F}$ being free and $\Omega$ satisfying (ii).
    \end{proof}

    As an example, consider the QRT of coherence (Sec.~\ref{sec:coherence}).
    The dephasing channel 
    \begin{equation}
        \Delta(\rho) = \sum_a \ket{a}\bra{a}\rho \ket{a}\bra{a},
    \end{equation}
    is a coherence-destroying channel for this QRT \cite{LH17}, i.e., $\Delta(\rho)\in\mathcal{F}(B)$, for any $\rho\in\mathcal{D}(A)$. 
    Due to $\Delta(\sigma)=\sigma$, for all $\sigma\in\mathcal{F}(A)$, it satisfies Def. \ref{def:RR} of a coherence-reducing channel.
    Any unitary that does not generate coherence [cf. Eq.~\eqref{eq:res-non-gen}] is of the form $U=\sum_{a=0}^{d-1} e^{i\phi_a}\ket{\pi(a)}\bra{a}$, where $\phi_a$ is a phase and $\pi$ denotes permutation of the incoherent basis $\ket{a}$ \cite{CG19}.
    A direct computation shows that
    \begin{equation}
        \begin{split}
            \Delta_{\Phi}(\rho) & = \big(\Phi^{-1}\circ\Delta\circ\Phi\big)(\rho),\\
            & = \sum_a  U^\dag\ket{a}\bra{a} U\rho  U^\dag\ket{a}\bra{a}U,\\
            & = \sum_a \ket{\pi(a)}\bra{\pi(a)}\rho \ket{\pi(a)}\bra{\pi(a)}=\Delta(\rho),
        \end{split}
    \end{equation}
    for any free unitary channel $\Phi(\rho)=U\rho U^\dag$.
    In particular, $\Delta$ is the unique coherence-destroying channel for this QRT.
    This provides an alternative understanding of the uniqueness of $\Delta$ established in Ref. \cite{G17}.


\begin{thebibliography}{XXX}

    \bibitem{K08}
        H. Kimble,
        The quantum internet, 
        \href{https://doi.org/10.1038/nature07127}{Nature \textbf{453}, 1023 (2008)}.

    \bibitem{RH25}
        P. P. Rohde, Z. Huang, Y. Ouyang, H.-L. Huang, Z.-E. Su, S. Devitt, R. Ramakrishnan, A. Mantri, S.-H. Tan, N. Liu, S. Harrison, C. Radhakrishnan, G. K. Brennen, B. Q. Baragiola, J. P. Dowling, T. Byrnes, W. J. Munro,
        The Quantum Internet (Technical Version),
        \href{https://doi.org/10.48550/arXiv.2501.12107}{arXiv.2501.12107 (2025)}.

    \bibitem{PS24}
        J. Pinske and J. Sperling,
        Unbreakable and breakable quantum censorship,
        \href{https://doi.org/10.1103/PhysRevA.109.052408}{Phys. Rev. A \textbf{109}, 052408 (2024)}.

    \bibitem{PM24}
        J. Pinske and K. M\o lmer,
        Censorship of quantum resources in quantum networks,
        \href{https://doi.org/10.1103/PhysRevA.110.062404}{Phys. Rev. A \textbf{110}, 062404 (2024)}.

    \bibitem{CG19}
        E. Chitambar and G. Gour, Quantum resource theories, 
        \href{https://doi.org/10.1103/RevModPhys.91.025001}{Rev. Mod. Phys. \textbf{91}, 025001 (2019)}.

    \bibitem{A06}
        J. Aberg,
        Quantifying Superposition,
        \href{https://doi.org/10.48550/arXiv.quant-ph/0612146}{arXiv:quant-ph/0612146}.

    \bibitem{BC14}
        T. Baumgratz, M. Cramer, and M. B. Plenio,
        Quantifying Coherence,
        \href{https://doi.org/10.1103/PhysRevLett.113.140401}{Phys. Rev. Lett. \textbf{113}, 140401 (2014)}.

    \bibitem{LM14}
        F. Levi and F. Mintert, 
        New J. Phys. \textbf{16}, 033007 (2014).

    \bibitem{SW18}
        J. Sperling and I. A. Walmsley,
        Quasiprobability representation of quantum coherence,
        \href{https://doi.org/10.1103/PhysRevA.97.062327}{Phys. Rev. A \textbf{97}, 062327 (2018)}.

    \bibitem{BR07}
        S. D. Bartlett, T. Rudolph, and R. W. Spekkens,
        Reference frames, superselection rules, and quantum information,
        \href{https://doi.org/10.1103/RevModPhys.79.555}{Rev. Mod. Phys. \textbf{79}, 555 (2007)}.

    \bibitem{GS08}
        G. Gour and R. W. Spekkens,
        The resource theory of quantum reference frames: manipulations and monotones,
        \href{https://doi.org/10.1088/1367-2630/10/3/033023}{New J. Phys. \textbf{10}, 033023 (2008)}.

    \bibitem{CA25}
        C. Cepollaro, A. Akil, P. Cie\'sli\'nski, A.-C. de la Hamette, and \v C. Brukner,
        Sum of Entanglement and Subsystem Coherence Is Invariant under Quantum Reference Frame Transformations,
        \href{https://doi.org/10.1103/h6b3-y4vt}{Phys. Rev. Lett. \textbf{135}, 010201 (2025)}.

    \bibitem{WK21}
        K.-D. Wu, T. V. Kondra, S. Rana, C. M. Scandolo, G.-Y. Xiang, C.-F. Li, G.-C. Guo, and A. Streltsov,
        Operational Resource Theory of Imaginarity,
        \href{https://doi.org/10.1103/PhysRevLett.126.090401}{Phys. Rev. Lett. \textbf{126}, 090401 (2021)}.

    \bibitem{XG21}
        S. Xue, J. Guo, P. Li, M. Ye, Y. Li,
        Quantification of resource theory of imaginarity,
        \href{https://doi.org/10.1007/s11128-021-03324-5}{Quantum Information Processing \textbf{20}, 383 (2021).}

    \bibitem{HH09}
        R. Horodecki, P. Horodecki, M. Horodecki, and K. Horodecki,
        Quantum entanglement,
        \href{https://doi.org/10.1103/RevModPhys.81.865}{Rev. Mod. Phys. \textbf{81}, 865 (2009)}.

    \bibitem{VS14}
        W. Vogel and J. Sperling
        Unified quantification of nonclassicality and entanglement,
        \href{https://doi.org/10.1103/PhysRevA.89.052302}{Phys. Rev. A \textbf{89}, 052302 (2014)}. 

    \bibitem{PAH24}
        J. Pinske, L. Ares, B. Hinrichs, . Kolb, J. Sperling,
        Separability Lindblad equation for dynamical open-system entanglement,
        \href{https://doi.org/10.1103/kd3b-bfxq}{Phys. Rev. A \textbf{113}, L010403 (2026)}.  
    
    \bibitem{OZ01}
        H. Ollivier and W. H. Zurek,
        Quantum Discord: A Measure of the Quantumness of Correlations,
        \href{https://doi.org/10.1103/PhysRevLett.88.017901}{Phys. Rev. Lett. \textbf{88}, 017901 (2001)}.

    \bibitem{MB12}
        K. Modi, A. Brodutch, H. Cable, T. Paterek, and V. Vedral,
        The classical-quantum boundary for correlations: Discord and related measures
        \href{https://doi.org/10.1103/RevModPhys.84.1655}{Rev. Mod. Phys. \textbf{84}, 1655 (2012)}.

    \bibitem{WPM09}
        S. Wu, U. V. Poulsen, and K. M\o lmer, 
        Correlations in local measurements on a quantum state, and complementarity as an explanation of nonclassicality, 
        \href{https://doi.org/10.1103/PhysRevA.80.032319}{Phys. Rev. A \textbf{80}, 032319 (2009)}.

    \bibitem{ES02}
        J. Eisert, S. Scheel, and M. B. Plenio,
        Distilling Gaussian States with Gaussian Operations is Impossible,
        \href{https://doi.org/10.1103/PhysRevLett.89.137903}{Phys. Rev. Lett. \textbf{89}, 137903 (2002)}.
        
    \bibitem{WP12}
        C. Weedbrook, S. Pirandola, R. Garc\'ia-Patr\'on, N. J. Cerf, T. C. Ralph, J. H. Shapiro, and S. Lloyd,
        Gaussian quantum information,
        \href{https://doi.org/10.1103/RevModPhys.84.621}{Rev. Mod. Phys. \textbf{84}, 621 (2012)}.

    \bibitem{R99}
        E. M. Rains, 
        Rigorous treatment of distillable entanglement,
        \href{https://doi.org/10.1103/PhysRevA.60.173}{Phys. Rev. A \textbf{60}, 173 (1999)}.

    \bibitem{P12}
        C. Palazuelos
        Superactivation of Quantum Nonlocality
        \href{https://doi.org/10.1103/PhysRevLett.109.190401}{Phys. Rev. Lett. \textbf{109}, 190401 (2012)}.

    \bibitem{PV22}
        C. Palazuelos and J. I. de Vicente,
        Genuine multipartite entanglement of quantum states in the multiple-copy scenario,
        \href{https://doi.org/10.22331/q-2022-06-13-735}{Quantum \textbf{6}, 735 (2022)}.	

    \bibitem{BK05}
        S. Bravyi and A. Kitaev,
        Universal quantum computation with ideal Clifford gates and noisy ancillas,
        \href{https://doi.org/10.1103/PhysRevA.71.022316}{Phys. Rev. A \href{71}, 022316 (2005)}.

    \bibitem{AC12}
        H. Anwar, E. T. Campbell and D. E. Browne,
        Qutrit magic state distillation,
        \href{https://doi.org/10.1088/1367-2630/14/6/063006}{New J. Phys. \textbf{14} 063006 (2012)}. 

    \bibitem{LP10}
        E.-M. Laine, J. Piilo, and H.-P. Breuer,
        Measure for the non-Markovianity of quantum processes
        \href{https://doi.org/10.1103/PhysRevA.81.062115}{Phys. Rev. A \textbf{81}, 062115 (2010)}.

    \bibitem{YZ20}
        Y. Yugra, F. D. Zela, and \'A. Cuevas,
        Coherence-based measurement of non-Markovian dynamics in an open quantum system,
        \href{https://doi.org/10.1103/PhysRevA.101.013822}{Phys. Rev. A \textbf{101}, 013822 (2020)}. 

    \bibitem{L61}
        R. Landauer, 
        Irreversibility and Heat Generation in the Computing Process,
        \href{https://doi.org/10.1147/rd.53.0183}{IBM J. Res. Dev. \textbf{5}, 183 (1961)}.

    \bibitem{GP05}
        B. Groisman, S. Popescu, and A. Winter,
        Quantum, classical, and total amount of correlations in a quantum state,
        \href{https://doi.org/10.1103/PhysRevA.72.032317}{Phys. Rev. A \textbf{72}, 032317 (2005)}. 

    \bibitem{JP99}
        D. Jonathan and M. B. Plenio,
        Entanglement-Assisted Local Manipulation of Pure Quantum States,
        \href{https://doi.org/10.1103/PhysRevLett.83.3566}{Phys. Rev. Lett. \textbf{83}, 3566 (1999)}. 

    \bibitem{CC10}
        L. Chen, E. Chitambar, R. Duan, Z. Ji, and A. Winter,
        Tensor Rank and Stochastic Entanglement Catalysis for Multipartite Pure States,
        \href{https://doi.org/10.1103/PhysRevLett.105.200501}{Phys. Rev. Lett. \textbf{105}, 200501 (2010)}. 

    \bibitem{DH03}
        W. v. Dam and P. Hayden,
        Universal entanglement transformations without communication
        \href{https://doi.org/10.1103/PhysRevA.67.060302}{Phys. Rev. A \textbf{67}, 060302(R) (2003)}.

    \bibitem{TS22}
        R. Takagi and N. Shiraishi,
        Correlation in Catalysts Enables Arbitrary Manipulation of Quantum Coherence
        \href{https://doi.org/10.1103/PhysRevLett.128.240501}{Phys. Rev. Lett. \textbf{128}, 240501 (2022)}. 

    \bibitem{VC02}
        G. Vidal and J. I. Cirac,
        Catalysis in Nonlocal Quantum Operations,
        \href{https://doi.org/10.1103/PhysRevLett.88.167903}{Phys. Rev. Lett. \textbf{88}, 167903 (2002)}.

    \bibitem{KE06}
        A. Kay and M. Ericsson,
        Local cloning of arbitrarily entangled multipartite states
        \href{https://doi.org/10.1103/PhysRevA.73.012343}{Phys. Rev. A 73, 012343 – Published 31 January, 2006} 

    \bibitem{BS99}
        C. H. Bennett, P. W. Shor, J. A. Smolin, and A. V. Thapliyal,
        Entanglement-Assisted Classical Capacity of Noisy Quantum Channels,
        \href{https://doi.org/10.1103/PhysRevLett.83.3081}{Phys. Rev. Lett. \textbf{83}, 3081 (1999)}.

    \bibitem{B02}
        G. Bowen,
        Entanglement required in achieving entanglement-assisted channel capacities,
        \href{https://doi.org/10.1103/PhysRevA.66.052313}{Phys. Rev. A \textbf{66}, 052313 (2002)}.

    \bibitem{HD08}
        M.-H. Hsieh, I. Devetak, and A. Winter,
        Entanglement-Assisted Capacity of Quantum Multiple-Access Channels,
        \href{https://doi.org/10.1109/TIT.2008.924726}{IEEE Trans. Inf. Theory \textbf{54}, 3078 (2008)}.

    \bibitem{W18}
        J. Watrous, 
        \textit{The Theory of Quantum Information}
        \href{https://doi.org/10.1017/9781316848142}{(Cambridge University Press, Cambridge, UK, 2018)}.

    \bibitem{LH17}
        Z.-W. Liu, X. Hu, and S. Lloyd,
        Resource Destroying Maps,
        \href{https://doi.org/10.1103/PhysRevLett.118.060502}{Phys. Rev. Lett. \textbf{118}, 060502 (2017)}.

    \bibitem{CL14}
        E. Chitambar, D. Leung, L. Mancinska, et al., 
        Everything You Always Wanted to Know About LOCC (But Were Afraid to Ask), 
        \href{https://doi.org/10.1007/s00220-014-1953-9}{Commun. Math. Phys. \textbf{328}, 303 (2014)}.

    \bibitem{EW00}
        J. Eisert and M. Wilkens,
        Quantum games,
        \href{https://doi.org/10.1080/09500340008232180}{J. Mod. Opt. \textbf{47}, 2543 (2000)}. 

    \bibitem{PSM25}
        J. Pinske, J. Sperling, K. M\o lmer,
        The entangling power of non-entangling channels,
        \href{https://doi.org/10.48550/arXiv.2512.14819}{arXiv:2512.14819 (2025)}.

    \bibitem{LW24}
        P. Lipka-Bartosik, H. Wilming, and N. H. Y. Ng,
        Catalysis in quantum information theory,
        \href{https://doi.org/10.1103/RevModPhys.96.025005}{Rev. Mod. Phys. \textbf{96}, 025005 (2024)}.

    \bibitem{DL16}
        J. M. Dominy, D. A. Lidar,
        Beyond complete positivity,
        \href{https://doi.org/10.1007/s11128-015-1228-1}{Quantum Inf. Process. \textbf{15}, 1349 (2016)}.

    \bibitem{P94}
        P. Pechukas,
        Reduced Dynamics Need Not Be Completely Positive,
        \href{https://doi.org/10.1103/PhysRevLett.73.1060}{Phys. Rev. Lett. \textbf{73}, 1060 (1994)}.  

    \bibitem{CV20} 
        E. Chitambar, J. I. de Vicente, M. W. Girard, and G. Gour, 
        Entanglement manipulation beyond local operations and classical communication, 
        \href{https://doi.org/10.1063/1.5124109}{J. Math. Phys. \textbf{61}, 042201 (2020)}.

    \bibitem{HN03}
        A. W. Harrow and M. A. Nielsen,
        Robustness of quantum gates in the presence of noise,
        \href{https://doi.org/10.1103/PhysRevA.68.012308}{Phys. Rev. A \textbf{68}, 012308 (2003)}.

    \bibitem{SV15}
        J. Sperling and W. Vogel,
        Convex ordering and quantification of quantumness,
        \href{https://doi.org/10.1088/0031-8949/90/7/074024}{Phys. Scr. \textbf{90}, 074024 (2015)}. 

    \bibitem{BPR00}
        C. H. Bennett, S. Popescu, D. Rohrlich, J. A. Smolin, and A. V. Thapliyal,
        Exact and asymptotic measures of multipartite pure-state entanglement,
        \href{https://doi.org/10.1103/PhysRevA.63.012307}{Phys. Rev. A \textbf{63}, 012307 (2000)}.

    \bibitem{N99}
        M. A. Nielsen,
        Conditions for a Class of Entanglement Transformations,
        \href{https://doi.org/10.1103/PhysRevLett.83.436}{Phys. Rev. Lett. \textbf{83}, 436 (1999)}.

    \bibitem{PM25}
        J. Pinske and K. M\o lmer,
        Retrodiction of measurement outcomes on a single quantum system revealing entanglement with its environment
        \href{https://doi.org/10.1103/kfn4-1fx4}{Phys. Rev. A \textbf{112}, 032208 (2025)}.

    \bibitem{HS03}
        M. Horodecki, P. W. Shor, M. B. Ruskai,
        General Entanglement Breaking Channels,
        \href{https://doi.org/10.1142/S0129055X03001709}{Rev. Math. Phys. \textbf{15}, 629 (2003)}. 

    \bibitem{SRB17}
        A. Streltsov, S. Rana, M. N. Bera, and M. Lewenstein,
        Towards Resource Theory of Coherence in Distributed Scenarios,
        \href{https://doi.org/10.1103/PhysRevX.7.011024}{Phys. Rev. X \textbf{7}, 011024 (2017)}. 
       
    \bibitem{B12}
        F. Buscemi,
        All Entangled Quantum States Are Nonlocal,
        \href{https://doi.org/10.1103/PhysRevLett.108.200401}{Phys. Rev. Lett. \textbf{108}, 200401 (2012)}. 

    \bibitem{BRS03}
        S. D. Bartlett, T. Rudolph, and R. W. Spekkens, 
        Classical and Quantum Communication without a Shared Reference Frame,
        \href{https://doi.org/10.1103/PhysRevLett.91.027901}{Phys. Rev. Lett. \textbf{91}, 027901 (2003)}.

    \bibitem{M97}
        K. M\o lmer,
        Optical coherence: A convenient fiction,
        \href{https://doi.org/10.1103/PhysRevA.55.3195}{Phys. Rev. A \textbf{55}, 3195 (1997)}.

    \bibitem{JB03}
        D. Janzing and T. Beth,
        Quasi-order of clocks and synchronism and quantum bounds for copying timing information
        \href{https://doi.org/10.48550/arXiv.quant-ph/0112138}{IEEE Trans. Inf. Theory \textbf{49}, 230 (2003)}.

    \bibitem{VAW08}
        J. A. Vaccaro, F. Anselmi, H. M. Wiseman, and K. Jacobs,
        Tradeoff between extractable mechanical work, accessible entanglement, and ability to act as a reference system, under arbitrary superselection rules,
        \href{https://doi.org/10.1103/PhysRevA.77.032114}{Phys. Rev. A \textbf{77}, 032114 (2008)}.

    \bibitem{GMS09}
        Gour, G., I. Marvian, and R. W. Spekkens,
        Measuring the quality of a quantum reference frame: The relative entropy of frameness,
        \href{https://doi.org/10.1103/PhysRevA.80.012307}{Phys. Rev. A \textbf{80}, 012307 (2009)}.

    \bibitem{G17}
        G. Gour,
        Quantum resource theories in the single-shot regime,
        \href{ttps://doi.org/10.1103/PhysRevA.95.062314}{Phys. Rev. A \textbf{95}, 062314 (2017)}.

    \bibitem{MS13}
        I. Marvian and R. W. Spekkens,
        The theory of manipulations of pure state asymmetry: I. Basic tools, equivalence classes and single copy transformations,
        \href{https://doi.org/10.1088/1367-2630/15/3/033001}{New J. Phys. \textbf{15} 033001 (2013)}.

    \bibitem{DHF21}
        F. Ding, X. Hu, and H. Fan,
        Amplifying asymmetry with correlating catalysts,
        \href{https://doi.org/10.1103/PhysRevA.103.022403}{Phys. Rev. A \textbf{103}, 022403 (2021)}.

    \bibitem{LWW23}
        L. van Luijk, R. F. Werner, and H. Wilming, 
        Covariant catalysis requires correlations and good quantum reference frames degrade little,
        \href{https://doi.org/10.22331/q-2023-11-06-1166}{Quantum \textbf{7}, 1166 (2023)}.

    \bibitem{BB07}
        H. Barnum, J. Barrett, M. Leifer, and A. Wilce,
        Generalized No-Broadcasting Theorem,
        \href{https://doi.org/10.1103/PhysRevLett.99.240501}{Phys. Rev. Lett. \textbf{99}, 240501 (2007)}.

    \bibitem{MS19}
        I. Marvian and R. W. Spekkens, 
        No-broadcasting theorem for quantum asymmetry and coherence and a trade-off relation for approximate broadcasting, 
        \href{https://doi.org/10.1103/PhysRevLett.123.020404}{Phys. Rev. Lett. \textbf{123}, 020404 (2019)}.

    \bibitem{SG24}
        J. Son, R. Ganardi, S. Minagawa, F. Buscemi, S. H. Lie, N. H. Y. Ng,
        Robust Catalysis and Resource Broadcasting: The Possible and the Impossible
        \href{https://doi.org/10.48550/arXiv.2412.06900}{arXiv.2412.06900 (2024)}.

    \bibitem{BDH06}
        T. Brun, I. Devetak, and M.-H. Hsieh,
        Correcting Quantum Errors with Entanglement,
        \href{https://doi.org/10.1126/science.1131563}{Science \textbf{314} 436 (2006).}

    \bibitem{WB09}
        M. M. Wilde and T. A. Brun,
        Protecting quantum information with entanglement and noisy optical modes,
        \href{https://doi.org/10.1007/s11128-009-0117-x}{Quantum Inf. Process. \textbf{8}, 401 (2009)}.

    \bibitem{LB13}
        C.-Y. Lai and T. A. Brun,
        Entanglement increases the error-correcting ability of quantum error-correcting codes,
        \href{https://doi.org/10.1103/PhysRevA.88.012320}{Phys. Rev. A \textbf{88}, 012320 (2013)}.    

    \end{thebibliography}
\end{document}